\documentclass[a4paper,UKenglish,cleveref, autoref,thm-restate]{lipics-v2021}

\hideLIPIcs  

\title{Skolem Meets Schanuel} 

\author{Yuri Bilu}%
{Institut de Math\'ematiques de Bordeaux, Universit\'e de Bordeaux and CNRS, Talence, France}%
{yuri@math.u-bordeaux.fr}{}{}

\author{Florian Luca}{
School of Mathematics, University of the Witwatersrand, Johannesburg, South Africa\and
Research Group in Algebraic Structures \& Applications, King Abdulaziz University,  Saudi Arabia
}{
}{https://orcid.org/0000-0003-1321-4422
}{
}

\author{Joris Nieuwveld}%
{Max Planck Institute for Software Systems, Saarland Informatics Campus, Germany}%
{jnieuwve@mpi-sws.org}{}{}

\author{Jo\"el Ouaknine}%
{Max Planck Institute for Software Systems, Saarland Informatics Campus, Germany}%
{joel@mpi-sws.org}{https://orcid.org/0000-0003-0031-9356}{%
Also affiliated with Keble College, Oxford as \href{http://emmy.network/}{\texttt{emmy.network}} Fellow.\\
DFG grant 389792660 as part of TRR 248 (see
\url{https://perspicuous-computing.science}).
}

\author{David Purser}{
University of Warsaw, Poland
}{%
}{https://orcid.org/0000-0003-0394-1634%
}{
}

\author{James Worrell}{
Department of Computer Science, University of Oxford, UK
}{jbw@cs.ox.ac.uk
}{https://orcid.org/0000-0001-8151-2443%
}{
}

\authorrunning{Y. Bilu, F. Luca, J. Nieuwveld, J. Ouaknine, D. Purser,
  and J. Worrell} 

\Copyright{Yuri Bilu, Florian Luca, Joris Nieuwveld, Jo\"el Ouaknine,
  David Purser, and James Worrell} 



\category{} 

\relatedversion{} 

\nolinenumbers 

\EventEditors{John Q. Open and Joan R. Access}
\EventNoEds{2}
\EventLongTitle{42nd Conference on Very Important Topics (CVIT 2016)}
\EventShortTitle{CVIT 2016}
\EventAcronym{CVIT}
\EventYear{2016}
\EventDate{December 24--27, 2016}
\EventLocation{Little Whinging, United Kingdom}
\EventLogo{}
\SeriesVolume{42}
\ArticleNo{23}

\usepackage{latexsym}
\usepackage{color}
\usepackage{tikz}
\usepackage{booktabs}
\newtheorem*{SC}{Skolem Conjecture}

\DeclareFontShape{T1}{lmr}{bx}{sc}{<->ssub*cmr/bx/sc}{}

\usepackage[textsize=tiny]{todonotes}

\begin{document}

\maketitle

\begin{abstract}
  The celebrated Skolem-Mahler-Lech Theorem states that the set of zeros of a
  linear recurrence sequence is the union of a finite set and finitely
  many arithmetic progressions.  The corresponding computational
  question, the Skolem Problem, asks to determine whether a given linear
  recurrence sequence has a zero term.  Although the Skolem-Mahler-Lech Theorem is
  almost 90 years old, decidability of the Skolem Problem remains
  open.  The main contribution of this paper is an algorithm to solve
  the Skolem Problem for simple linear recurrence sequences (those
  with simple characteristic roots).  Whenever the algorithm
  terminates, it produces a stand-alone certificate that its output is
  correct---a set of zeros together with a collection of witnesses that no further
  zeros exist. We give a proof that the algorithm always terminates assuming two
  classical number-theoretic conjectures: the Skolem Conjecture (also
  known as the Exponential Local-Global Principle) and the $p$-adic
  Schanuel Conjecture. Preliminary experiments with an implementation
  of this algorithm within the tool \textsc{Skolem} point to the practical applicability of this method. 
\end{abstract}

\section{Introduction}
\subsection{The Skolem Problem}
A linear recurrence sequence (LRS) $\boldsymbol{u}=\langle u_n\rangle_{n=0}^\infty$
is a sequence of rational numbers satisfying the equation
\begin{equation}
\label{eq:RECUR}
u_{n+d} = c_1u_{n+d-1} + \cdots + c_{d-1}u_{n+1}+ c_du_{n} 
\end{equation}
for all $n \in \mathbb{N}$, where the coefficients $c_1,\ldots,c_d$
are rational numbers and $c_d \neq 0$.  We say that the above
recurrence has \emph{order} $d$.  We moreover say that an LRS is \emph{simple} if the
characteristic polynomial of its minimal-order recurrence has simple
roots.

The celebrated theorem of Skolem, Mahler, and Lech (see~\cite{Book})
describes the structure of the set $\{ n \in \mathbb{N} : u_n = 0\}$
of zero terms of an LRS as follows:
\begin{theorem}
  \label{thm:SML}
  Given a linear recurrence sequence $\boldsymbol{u}=\langle u_n\rangle_{n=0}^\infty$,
  the set of zero terms is a union of finitely many arithmetic
  progressions, together with a finite set.
\end{theorem}

The statement of Theorem~\ref{thm:SML} can be refined by considering  
the notion of \emph{non-degeneracy} of an LRS\@.  An LRS is  
non-degenerate if in its minimal recurrence the quotient of no two  
distinct roots of the characteristic polynomial is a root of unity.  A  
given LRS can be effectively decomposed as an interleaving of finitely  
many non-degenerate sequences, some of which may be identically zero.  
The core of the Skolem-Mahler-Lech Theorem is the fact that a non-zero  
non-degenerate linear recurrence sequence has finitely many zero
terms.  Unfortunately, all known proofs of this last assertion are
ineffective: it is not known how to compute the finite set of zeros of
a given non-degenerate linear recurrence sequence.  It is readily seen
that the existence of a procedure to do so is equivalent to the existence
of a procedure to decide whether an arbitrary given LRS has a zero
term.  The problem of deciding whether an LRS has a zero term is
variously known as the  Skolem Problem or the Skolem-Pisot Problem.

Decidability of the Skolem Problem is known only for certain special
cases, based on the relative order of the absolute values of the
characteristic roots.  Say that a characteristic root $\lambda$ is
\emph{dominant} if its absolute value is maximal among all the
characteristic roots.  Decidability is known in case there are at most
$3$ dominant characteristic roots, and also for recurrences of order
at most $4$~\cite{MST84,Ver85}.  However for LRS of order $5$ it is not
currently known how to decide the Skolem Problem.

The Skolem Problem, along with closely related questions such as the
Positivity Problem, is intimately connected to various fundamental
topics in program analysis and automated verification, such as the
termination and model checking of simple while
loops~\cite{AKK21,KLO22,OuaknineW15} or the algorithmic analysis of
stochastic systems~\cite{AAO15,AgrawalAGT15,BJK20,CD21,PB20}.  It also
appears in a variety of other contexts, such as formal power
series~\cite{RS94,Soi76} and control theory~\cite{BlondelT00,FOP19}.
The Skolem Problem is often used as a reference to establish hardness
of other open decision problems; in addition to some of the previously
cited papers, the articles~\cite{BFJ21,COW15}, for example,
specifically invoke hardness of the Skolem Problem for simple LRS of order $5$.
Thus far, the only known complexity bound for the Skolem Problem is
NP-hardness~\cite{BP02}.
\subsection{The Skolem Conjecture and the Bi-Skolem Problem}
The notion of linear recurrence equally well makes sense for a
bi-infinite sequence
$\boldsymbol u = \langle u_n \rangle_{n=-\infty}^\infty$ of rational
numbers: one defines $\boldsymbol u$ to be a \emph{linear recurrent
  bi-sequence (LRBS)} if it satisfies the recurrence~\eqref{eq:RECUR}
for all $n\in \mathbb{Z}$.  Note that every LRS $\boldsymbol u$
extends uniquely to an LRBS satisfying the same recurrence (one obtains
such an extension by ``running the recurrence backwards'').  The
notions of simplicity and non-degeneracy carry over in the obvious way
to LRBS\@.  We remark also that the Skolem-Mahler-Lech Theorem remains
valid for LRBS---a non-degenerate LRBS has finitely many zeros.  The
analog of the Skolem Problem for LRBS is the \emph{Bi-Skolem Problem},
which asks, for a given LRBS
$\boldsymbol u = \langle u_n \rangle_{n=-\infty}^\infty$, whether there exists
$n\in \mathbb{Z}$ with $u_n=0$.

A major motivation to consider the Bi-Skolem Problem is the existence
of the \emph{Exponential Local-Global Principle}, a conjecture
introduced by Thoralf Skolem in 1937~\cite{Sko37}.  To formulate the conjecture we
first make some observations about the value set of an LRBS\@.  Given a
non-zero integer $b$, let $\mathbb{Z}[\frac{1}{b}]$ be the
subring of $\mathbb{Q}$ obtained by adjoining $\frac{1}{b}$ to
$\mathbb{Z}$.  We note that every rational LRBS takes values
in $\mathbb{Z}[\frac{1}{b}]$ for some $b$.  Indeed, if
$\boldsymbol u=\langle u_n\rangle_{n=-\infty}^{\infty}$ satisfies
recurrence~\eqref{eq:RECUR} and $\mathbb{Z}[\frac{1}{b}]$
contains the coefficients $c_1,\ldots,c_d$, the reciprocal $c_d^{-1}$ of the last
coefficient, and the initial values $u_0,\ldots,u_{d-1}$, then by
running the recurrence forwards and backwards from the initial values
we see that $u_n \in \mathbb{Z}[\frac{1}{b}]$ for all
$n\in\mathbb{Z}$.


\begin{SC}
  \label{conj1}
  Let $\boldsymbol u$ be a simple rational LRBS taking values in the
  ring $\mathbb{Z}[\frac{1}{b}]$ for some integer $b$.
  Then $\boldsymbol u$ has no zero iff, for some integer $m \geq 2$
  with $\gcd(b,m)=1$, we have that $u_n \not\equiv 0 \bmod{m}$ for all
  $n \in \mathbb{Z}$.
\end{SC}

In other words, the Skolem Conjecture asserts that if a simple LRBS
fails to have a zero, then this is witnessed modulo $m$ for some $m$.
The truth of this conjecture immediately entails the existence of an algorithm to
solve the Bi-Skolem Problem for simple LRBS\@: simply search in parallel
either for a zero of the LRBS, or for a number $m$ substantiating the
absence of zeros. If the Skolem Conjecture holds, then the search must
necessarily terminate in finite time.

There exists a substantial body of literature on the Skolem
Conjecture, including proofs of a variety of special cases. In
particular, the Skolem Conjecture has been shown to hold for simple
LRBS of order 2~\cite{BBL13}, and for certain families of
LRBS of order 3~\cite{Sch77,Sch03}. In a different but
related vein, Bert\'ok and Hajdu have shown that, in some sense, the
Skolem Conjecture is valid in ``almost all''
instances~\cite{BH16,BH18}.

\subsection{Main Results}
It is immediate that the Bi-Skolem Problem reduces to the Skolem
Problem: an LRBS $\langle u_n \rangle_{n=-\infty}^\infty$ has a zero
term if and only if at least one of the one-way infinite sequences
$\langle u_n\rangle_{n=0}^\infty$ and
$\langle u_{-n}\rangle_{n=0}^\infty$, both of which are LRS, has a
zero term.  However it is open whether there is a reduction in the
other direction (equivalently, it is open whether an oracle for the
Bi-Skolem Problem can be used to determine \emph{all} the zeros of a
non-degenerate LRBS)\@.  Indeed, an oracle for the Bi-Skolem Problem
would appear to be of little utility in deciding the Skolem Problem
for an LRS whose bi-completion happens to harbour a zero at a negative
index.  It is likewise not known (in spite of the 
similar nomenclature) whether the truth of the Skolem Conjecture implies
decidability of the Skolem Problem.

Our first main result is as follows:
\begin{theorem}
  The Skolem Problem reduces to the Bi-Skolem Problem subject to the 
  weak $p$-adic Schanuel Conjecture. 
  \label{thm:main}
\end{theorem}
Schanuel's Conjecture~\cite[Pages 30-31]{Lang66} is a unifying 
conjecture in transcendental number theory that plays a key role in 
the study of the exponential function over both the real and complex 
numbers.  In particular, a celebrated paper of Macintyre and 
Wilkie~\cite{MacintyreW96} obtains decidability of the first-order 
theory of the structure $(\mathbb{R};<,\,\cdot\,,+,\exp)$ assuming 
Schanuel's Conjecture over $\mathbb{R}$.  A $p$-adic version of the 
Schanuel Conjecture, referring to the exponential function on the ring 
$\mathbb{Z}_p$ of $p$-adic integers, was formulated in~\cite{CM09}. 
This conjecture was shown in~\cite{Mariaule} to imply decidability of 
the first-order theory of the structure 
$(\mathbb{Z}_p;<,\,\cdot\,,+,\exp)$.

Since the reduction in Theorem~\ref{thm:main} specialises to simple
LRBS we obtain:
\begin{theorem}
  The Skolem Problem for simple LRS is decidable subject to the weak $p$-adic 
  Schanuel Conjecture and the Skolem Conjecture. 
\label{thm:main2}
\end{theorem}
The proof of Theorem~\ref{thm:main2} gives an algorithm that computes
the set of zeros of a non-degenerate simple LRBS\@.  The algorithm
moreover produces an unconditional certificate that its output is
correct, i.e., that all zeros have been found.  This certificate
consists of a partition of the input LRBS into finitely many
subsequences such that each subsequence contains at most one zero.
For a subsequence with no zero, the algorithm finds an integer $m$
such that the subsequence is non-zero modulo $m$; for a subsequence
with a zero, the algorithm provides a prime $p$ and a subsequence of
this subsequence with exactly one zero as $p$ divides the other terms a
well-described (bounded) number of times.  The conjectural aspect of
Theorem~\ref{thm:main2} solely concerns the proof that the algorithm
terminates on all input sequences.

We have implemented our algorithm within the \textsc{Skolem}
tool,\footnote{\textsc{Skolem} can be experimented with online at
  \texttt{https://skolem.mpi-sws.org/} .}
which enumerates the set of zeros of a given non-degenerate simple
LRS, and produces an independent (conjecture-free) certificate that all zeros have been
found. Preliminary experiments, which we present in
Section~\ref{tool}, point to the practical applicability of our algorithm.

\subsection{Related Work}

The decidability of the Skolem Problem is generally considered to have
been open since the early 1930s, as the $p$-adic techniques underpinning
the Skolem-Mahler-Lech Theorem were well understood already at the time
not to be effective. As noted earlier, a breakthrough establishing
decidability at order $4$ occurred in the
mid-1980s~\cite{MST84,Ver85}, making key use of Baker's theorem on
linear forms in logarithms of algebraic numbers. (It is perhaps worth
noting that the attendant algorithms are however unlikely to be implementable in
practice, as the constants involved are huge.) Very recently, we have
shown that the Skolem Problem is decidable at order $5$ assuming only
the Skolem Conjecture; and in the same paper we also obtained unconditional decidability
for reversible LRS\footnote{An integer LRS is \emph{reversible} if its
  completion as an LRBS only takes on integer values.} of order $7$ or
less~\cite{ LiptonLNOPW22}. In the present paper, we improve on the
former result by establishing a Turing reduction from the Skolem Problem
at order $5$ to the Bi-Skolem Problem for simple LRBS of order $5$;
this is the content of Theorem~\ref{thm:order-5}.

\section{Technical Background}
\label{sec:background}
\subsection{Computation in Number Fields}
\label{sec:symbolic}
A number field $\mathbb{K}$ is a finite-degree extension of
$\mathbb{Q}$.  For computational purposes, such a field can be
represented in the form $\mathbb{Q}[X]/(g(X))$, where $g(X)$ is the
minimal polynomial of a primitive element of $\mathbb{K}$.  With such a
representation it is straightforward to do arithmetic in $\mathbb{K}$,
including solving systems of linear equations with coefficients in
$\mathbb{K}$.  Moreover, given a polynomial $f(X) \in \mathbb{Q}[X]$,
one can compute a representation in the above form of the splitting
field $\mathbb{K}$ of $f$ over $\mathbb{Q}$, together with
representations of the roots of $f$ as elements of
$\mathbb{K}$~\cite{Landau}.


In addition to basic arithmetic and linear algebra in $\mathbb{K}$, we
wish to determine whether some given elements
$\lambda_1,\ldots,\lambda_s \in \mathbb{K}$ are multiplicatively
independent and, if not, to exhibit $a_1,\ldots,a_s \in \mathbb{Z}$
such that $\lambda_1^{a_1} \cdots \lambda_m^{a_s}=1$.  For this we can
use the following result, which shows that if such a multiplicative
relation exists then there exists one in which the exponents
$a_1,\ldots,a_m$ have absolute value at most $B$ for some bound $B$
computable from the height of the $\lambda_i$ and the degree of the
number field $\mathbb{K}$.

\begin{theorem}[Masser~\cite{Mas88}]\label{th:masser}
  Let $\mathbb{K}$ be a number field of degree $D$ over $\mathbb{Q}$.
  For $s\geq 1$ let $\lambda_1,\ldots,\lambda_s$ be non-zero elements
  of $\mathbb{K}$ having height at most $h$ over $\mathbb{Q}$.  Then
  the group of multiplicative relations
        \begin{gather}
       L=\{(k_1,\ldots,k_s)\in \mathbb{Z}^s : \lambda_1^{k_1}\cdots \lambda_s^{k_s} = 1\}
       \label{eq:lattice}   
        \end{gather} 
        is generated (as an additive subgroup of $\mathbb{Z}^s$) by
            a collection of vectors all of whose entries have absolute
            value at most
        \[ (csh)^{s-1}D^{s-1}\frac{(\log(D+2))^{3s-3}}{(\log\log(D+2))^{3s-4}} \, ,\]
        for some absolute constant $c$.
        \label{the:masser}
\end{theorem}

\subsection{$p$-adic Numbers}
\label{sec:p-adic}
Let $p$ be a prime.  Define the \emph{$p$-adic valuation}
$v_p : \mathbb{Q} \rightarrow \mathbb{Z} \cup \{ \infty \}$ by
$v_p(0)=\infty$ and $v_p\left(p^\nu \cdot \frac{a}{b}\right) = \nu$
for all $a,b \in \mathbb{Z}\setminus\{0\}$ such that $\gcd(ab,p)=1$.
In other words, $v_p(x)$ gives the exponent of~$p$ as a divisor of
$x\in\mathbb{Q}$.  The map $v_p$ determines an absolute value
$|\cdot |_p$ on $\mathbb{Q}$, where $|x|_p:=p^{-v_p(x)}$ (with the
convention that $|0|_p=p^{-\infty}=0$).  Due to the fact that
$v_p(a+b) \geq \min(v_p(a),v_p(b))$, we have the strong triangle
equality: $|a+b|_p \leq \max(|a|_p,|b|_p)$ for all
$a,b \in \mathbb{Q}$.  In other words, $|\cdot|_p$ is a
\emph{non-Archimedean} absolute value.  The field $\mathbb{Q}_p$ of
\emph{$p$-adic numbers} is the completion of $\mathbb{Q}$ with respect
to $|\cdot |_p$.  The absolute value $|\cdot |_p$ extends to a
non-Archimedean absolute value on $\mathbb{Q}_p$.  The
ring of \emph{$p$-adic integers} is
$\mathbb{Z}_p := \{ x \in \mathbb{Q}_p : |x|_p \leq 1\}$.  The ring
$\mathbb{Z}_p$ contains a unique maximal ideal $p\mathbb{Z}_p$, with
the quotient $\mathbb{Z}_p/p\mathbb{Z}_p$ being isomorphic to
$\mathbb{F}_p$ (the finite field with $p$ elements).  When we refer to
elements of $\mathbb{Z}_p$ modulo $p$ we refer to their image under
this quotient map.

Given a sequence of numbers $\langle a_n \rangle_{n=0}^\infty$ in $\mathbb{Z}_p$,
the infinite sum $\sum_{n=0}^\infty a_n$ converges to an element of
$\mathbb{Z}_p$ if and only if $|a_n|_p \rightarrow 0$ (equivalently,
$v_p(a_n) \rightarrow \infty$) as $n\rightarrow \infty$.  It follows
that given a
sequence $\langle a_n \rangle_{n=0}^\infty$ in $\mathbb{Z}_p$ with
$|a_n|_p \rightarrow 0$, the corresponding power series
$f(X)=\sum_{j=0}^\infty a_j X^j$ defines a function
$f:\mathbb{Z}_p\rightarrow \mathbb{Z}_p$.

Consider a monic polynomial $g(X) \in \mathbb{Z}[X]$ with non-zero
discriminant $\Delta(g)$.  Let $p$ be a prime that does not divide
$\Delta(g)$.
Denote by $\overline{g}(X) \in \mathbb{F}_p[X]$ the
polynomial obtained from $g$ by replacing each coefficient with its
residue modulo $p$.  It is well known that a sufficient condition for
$g$ to split completely over $\mathbb{Z}_p$ is that $\overline{g}$
split over $\mathbb{F}_p$. Indeed, in this situation one can use
Hensel's Lemma~\cite[Theorem 3.4.1]{gouvea} to ``lift'' each of the
roots of $\overline{g}$ in $\mathbb{F}_p$ to a distinct root in
$\mathbb{Z}_p$.  Moreover, by the Chebotarev density theorem~\cite{Lagarias}
there are infinitely many primes $p$ for which $\overline{g}$ splits
over $\mathbb{F}_p$.  Hence there are infinitely many primes $p$ such
that $g$ splits over $\mathbb{Z}_p$.  Note that the last statement
holds even without the assumption that $\Delta(g)\neq 0$, since
$g\in\mathbb{Z}[X]$ splits over $\mathbb{Z}_p$ whenever
$\frac{g}{\gcd(g,g')} \in \mathbb{Z}[X]$ splits over $\mathbb{Z}_p$
(and the latter has non-zero discriminant).

Let $p$ be an odd prime.\footnote{We omit the prime $p=2$ to avoid
  special cases in the facts below.}  The \emph{$p$-adic exponential}
is defined by the power series
\[ \exp(x) = \sum_{k=0}^\infty \frac{x^k}{k!}  \, , \]
which converges for all $x \in p\mathbb{Z}_p$.
The \emph{$p$-adic logarithm} 
is defined by
the power series
\[ \log(x) = \sum_{k=0}^\infty (-1)^{k+1} \frac{(x-1)^k}{k} \, , \]
which converges for all $x \in 1+p\mathbb{Z}_p$.  For
$x,y \in p\mathbb{Z}_p$ we have $\exp(x+y)=\exp(x)\exp(y)$ and for
$x,y \in 1+p\mathbb{Z}_p$ we have $\log(xy)=\log(x)+\log(y)$.  Indeed
we have that $\exp$ and $\log$ yield mutually inverse isomorphisms
between the additive group $p\mathbb{Z}_p$ and multiplicative group
$1+p\mathbb{Z}_p$.

Schanuel's Conjecture for the complex numbers is a powerful unifying
principle in transcendence theory.  We will need the following
$p$-adic version of the weak form of Schanuel's Conjecture, which can
be found, e.g., as~\cite[Conjecture 3.10]{CM09}.
\begin{conjecture}(Weak $p$-adic Schanuel Conjecture.)  Let
  $\alpha_1,\ldots,\alpha_s \in 1+p\mathbb{Z}_p$ be such that 
  $\log \alpha_1,\ldots,\log \alpha_s$ are linearly independent over 
  $\mathbb{Q}$.  Then $\log \alpha_1,\ldots,\log \alpha_s$ are
  algebraically independent over $\mathbb{Q}$, that is, for every 
  non-zero polynomial $P \in \mathbb{Q}_p [X_1,\ldots,X_s]$ whose
  coefficients are algebraic over $\mathbb{Q}$,
  it holds that
  $P(\log \alpha_1,\ldots,\log \alpha_s) \neq 0$.
\label{con:schanuel}
\end{conjecture}

A known special case of 
Conjecture~\ref{con:schanuel} is the following result of
Brumer~\cite{brumer}, which is a $p$-adic analog of Baker's Theorem on
linear independence of logarithms of algebraic numbers.
\begin{theorem}
  \label{thm:p-adic-baker}
  Let $\alpha_1,\ldots,\alpha_s \in 1+p\mathbb{Z}_p$ be algebraic over
  $\mathbb{Q}$ and such that $\log \alpha_1,\ldots,\log \alpha_s$ are
  linearly independent over $\mathbb{Q}$.  Then
  $\beta_0+\beta_1\log \alpha_1+\cdots +\beta_s\log \alpha_s\neq 0$
  for all $\beta_0,\ldots,\beta_s \in \mathbb{Q}_p$ that are algebraic
  over $\mathbb{Q}$ and not all zero.
\end{theorem}

\section{$p$-adic Power-Series Representation of LRBS}
\label{sec:power}
Let $\boldsymbol u = \langle u_n \rangle^{\infty}_{n=-\infty}$ be an
LRBS of rational numbers  satisfying the linear recurrence
\begin{gather*}
  u_{n+d} = c_1 u_{n-d-1} + \cdots + c_d u_{n} \quad(n\in\mathbb{Z}),
\end{gather*}
where $c_d\neq 0$.  For the purposes of computing the zeros of
$\boldsymbol u$ we can assume without loss of generality that the
coefficients $c_1,\ldots,c_d$ of the recurrence are integers.
(It is easy to see that for any integer $\ell$ such that
$\ell c_i \in \mathbb{Z}$ for $i\in \{1,\ldots,d\}$, the scaled
sequence $\langle \ell^n u_n \rangle_{n=-\infty}^\infty$ satisfies
an integer recurrence.)  Write $g(X):=X^d - c_1X^{d-1}-\cdots - c_d$
for the \emph{characteristic polynomial} of the recurrence and write
\[ A := \begin{pmatrix}
    c_1 &  \cdots & c_{d-1}& c_d \\
    1    & \cdots   & 0         & 0 \\
    \vdots &         & \vdots   & \vdots \\
    0 &      \cdots             & 1 & 0 
  \end{pmatrix} \] for the \emph{companion matrix}.
Setting
\begin{gather*}
  \alpha:=\begin{pmatrix} 0 & \cdots & 0 &1\end{pmatrix}
  \quad\text{and}\quad \beta := \begin{pmatrix} u_{d-1}& \cdots & u_1
    & u_0\end{pmatrix}^T  ,
\end{gather*}
we have the matrix-exponential representation 
$u_n = \alpha A^n \beta$ for all $n\in \mathbb{Z}$.

The key tool in our approach---which also underlies the proof of the
Skolem-Mahler-Lech Theorem---is the representation of the LRBS
$\boldsymbol{u}$ in terms of a power series
$f(X)=\sum_{j=0}^\infty a_jX^j$ with coefficients in $\mathbb{Z}_p$.
In defining $f$ we work with an odd prime $p$ such that (i)~$p$ does
not divide the constant term $c_d$ of the recurrence; (ii)~$p$ does
not divide the discriminant $\Delta\left(\frac{g}{\gcd(g,g')}\right)$;
(iii)~the characteristic polynomial $g$ splits over $\mathbb{Z}_p$.
As explained in Section~\ref{sec:p-adic}, there are infinitely many
such primes.  Moreover, for a particular prime $p$ that does not
divide $\Delta\left(\frac{g}{\gcd(g,g')}\right)$,
we can easily
verify whether $g$ splits over $\mathbb{Z}_p$, since this is
equivalent to $\frac{g}{\gcd(g,g')}$ splitting over $\mathbb{F}_p$.

Write $\lambda_1,\ldots,\lambda_s \in \mathbb{Z}_p$ for the distinct
roots of $g$.
Let $\mathbb{K}$ be the subfield of $\mathbb{Q}_p$ generated by
$\lambda_1,\ldots,\lambda_s$.  Then $\mathbb{K}$ is a number field and
thus we can compute symbolically in $\mathbb{K}$ as described in
Section~\ref{sec:symbolic}.  It is well known that the sequence
$\boldsymbol{u}$ admits an exponential polynomial representation
\begin{gather}
  u_n =\sum_{i=1}^s Q_i(n) \lambda_i^n \quad(n \in \mathbb{Z}) \, 
  ,
\label{eq:EP}
\end{gather}
where $Q_i \in \mathbb{K}[X]$ has degree strictly less than the
multiplicity of $\lambda_i$ as a root of $g$.  The coefficients of
each polynomial $Q_i$ can be computed as the unique solution of the
system of linear equations that arises by substituting $n=0,\ldots,d-1$ in
Equation~\eqref{eq:EP} (where, recall, $d$ is the order of the
recurrence).

The companion matrix has determinant $\det(A)=\pm c_d$, which is 
non-zero modulo $p$; hence $A$ is invertible modulo $p$.  Let $L$ be 
the least positive integer such that $A^L\equiv I \bmod p$.  Being an eigenvalue of $A^L$, $\lambda_i^L \equiv 1 \bmod p$ for all 
$i\in \{1,\ldots,s\}$ and hence 
the $p$-adic logarithm $\log \lambda^L_i$ is defined
for all $i\in \{1,\ldots,s\}$.
We thus obtain the following representation of
the subsequence $\langle u_{Ln}\rangle_{n=-\infty}^\infty$
in terms of the $p$-adic exponential and logarithm functions:
\[
  u_{Ln} =  \sum_{i=1}^s Q_i(Ln) \lambda_i^{Ln} =
                                                          \sum_{i=1}^s Q_i(Ln) \exp(n \log 
      \lambda^L_i)  \,   .
\]
  Now consider the power series $f(X) = \sum_{j=0}^\infty a_j X^j$
  such that
  \begin{gather}
  f(x) := \sum_{i=1}^s Q_i(Lx) \exp( x \log \lambda_i^L)  
\label{eq:exp-log}
\end{gather}
for all $x \in \mathbb{Z}_p$.  Then we have 
    $u_{Ln} = f(n)$ for all $n\in \mathbb{Z}$.
    Moreover, since $ a_j = \frac{1}{j!}f^{(j)}(0)$, by taking derivatives in
    \eqref{eq:exp-log} we obtain the following expression for the
    coefficients of $f$:
    \begin{gather} a_j = \frac{1}{j!} \sum_{i=1}^s \sum_{k=0}^j
    \binom{j}{k}   L^k Q_i^{(k)} (0) \left(\log 
    \lambda_i^L\right)^{j-k}  \, . 
\label{eq:polyF}
\end{gather}

In the remainder of this section we give an alternative formula for
the coefficients of $f$ as $p$-adically convergent sums of rational
numbers.  This provides a simple method to compute $v_p(a_j)$ that
avoids computing $p$-adic approximations of the characteristic roots,
as would be needed if we were to directly use~\eqref{eq:polyF}.

Recall that we have $A^L\equiv I \mod p$.  Let us say that
$A^L=I+pB$ for some integer matrix $B$.  Then we have:
\begin{eqnarray*}
u_{Ln} &=& \alpha A^{Ln} \beta \\
    &=& \alpha (I+pB)^n \beta \\
       &=& \sum_{k=0}^n \binom{n}{k} p^k \alpha B^k \beta \\
   &= & \sum_{k=0}^n  \frac{n(n-1)\ldots (n-k+1)}{k!} p^k \alpha B^k 
        \beta\\
  &=& \sum_{k=0}^\infty \frac{n(n-1)\ldots (n-k+1)}{k!} p^k  \alpha B^k 
      \beta\\
  &=& \sum_{k=0}^\infty \sum_{j=0}^\infty  c_{k,j}n^j \frac{p^k}{k!}
      \quad\mbox{for certain $c_{k,j} \in \mathbb{Z}$ with $c_{k,j}=0$
      for $j>k$}\\
       &=& \sum_{j=0}^\infty \sum_{k=j}^\infty c_{k,j}n^j
           \frac{p^k}{k!} \, .
\end{eqnarray*}
It remains to see why one can reverse the order of summation in the
last line above and why the resulting sums converge in $\mathbb{Z}_p$.
For this we can apply~\cite[Proposition 4.1.4]{gouvea}, for which we
require that the summand $c_{k,j}n^j 
  \frac{p^k}{k!}$ converge to $0$ as $j\rightarrow \infty$ and 
  converge to $0$ uniformly in $j$ as $k \rightarrow \infty$. 
  But this precondition follows from the fact that
$v_p(k!) < \frac{k}{p-1}$, from which we have
  $v_p\left(c_{k,j}n^j\frac{p^k}{k!} \right) \geq \frac{(p-2)k}{p-1}$ for all
  $k\geq j$.   

 Now consider the power series $\widetilde{f}(X):=\sum_{j=0}^\infty
  b_jX^j$ where 
\begin{gather}
    b_j := \sum_{k=j}^\infty 
    c_{k,j}\frac{p^k}{k!} \in \mathbb{Z}_p \, .
\label{eq:power-series}
  \end{gather}
  By the above considerations we have that
  $v_p(b_j) \geq \frac{(p-2)j}{p-1}$ and hence $\widetilde{f}$
  converges on $\mathbb{Z}_p$ and satisfies $\widetilde{f}(n)=u_{Ln}$
  for all $n\in \mathbb{Z}$.  In particular, the power series $f$ and
  $\widetilde{f}$ agree on $\mathbb{Z}$ and hence (e.g.,
  by~\cite[Proposition 4.4.3]{gouvea}) are identical, i.e., $a_j=b_j$
  for all $j\in \mathbb{N}$.  Thus we can use
  Equation~\eqref{eq:power-series} to exactly compute $v_p(a_j)$ for
  any $j$ such that $a_j\neq 0$.

  \section{Computing all the Zeros of an LRBS}
In this section we show, assuming the weak $p$-adic Schanuel
Conjecture, that the set of all zeros of a non-degenerate 
LRBS is computable using  an oracle for the Bi-Skolem Problem.  In
particular, this gives a Turing reduction of the Skolem Problem to the
Bi-Skolem Problem.

\begin{proposition}
  Let $f : \mathbb{Z}_p\rightarrow \mathbb{Z}_p$ be given by a
  convergent $p$-adic power series $f(X)=\sum_{k=0}^\infty a_k X^k$,
  with coefficients in $\mathbb{Z}_p$.  Suppose that $\ell$ is a positive
  integer such that $a_0=\cdots=a_{\ell-1}=0$ and $a_\ell \neq 0$.
  Then, writing $\nu:=v_p(a_\ell)$, we have $f(p^{\nu+1}x) \neq 0$ for all
  non-zero $x \in \mathbb{Z}_p$.
\label{prop:zero}
\end{proposition}
\begin{proof}
  Let $x \in \mathbb{Z}_p$ be non-zero.  For every $m>\ell$ we have
  \begin{eqnarray*}
v_p(a_\ell (p^{\nu+1}x)^\ell) 
                            &=& \nu+ \ell(\nu+1) + v_p(x^\ell) \\
                            &< & m(\nu+1) + v_p(x^m) \quad\mbox{(since
                                 $\ell<m$ and $x\neq 0$)}\\
    &\leq & v_p(a_m(p^{\nu+1}x)^m) \, .
  \end{eqnarray*}
  It follows that for all $m \geq \ell$, 
\[  v_p\left( \sum_{k=0}^m a_k (p^{\nu+1}x)^k \right) =
                                                            v_p(a_\ell
                                                            (p^{\nu+1}x)^\ell)
                                                            \, . \]
Letting $m$ tend to infinity, we have                                                          
$v_p(f(p^{\nu+1}x)) = v_p(a_\ell (p^{\nu+1}x)^\ell) <
  \infty$ and we conclude that $f(p^{\nu+1}x) \neq 0$.
  \end{proof}

  \begin{proposition}
    Let $\boldsymbol u = \langle u_n \rangle_{n=-\infty}^\infty$ be a
    non-zero LRBS consisting of rational numbers.  Assuming the weak
    $p$-adic Schanuel Conjecture, one can compute a positive integer
    $M$ such that $u_{Mn} \neq 0$ for all
    $n \in \mathbb{Z}\setminus \{0\}$.
\label{prop:oracle}
\end{proposition}
\begin{proof}
  As explained in Section~\ref{sec:power}, there exists a prime $p$
  and a positive integer $L$ such that $u_{Ln} = f(n)$ for all
  $n \in \mathbb{Z}$, for the $p$-adic power series $f(X)=\sum_{j=0}^\infty a_j X^j$
  whose coefficients are given by the
  formula~\eqref{eq:exp-log}.  Recall that in this formula the
  $\lambda_i$ are the characteristic roots of $\boldsymbol u$ and the
  $Q_i$ are the coefficients appearing in the exponential polynomial
  formula~\eqref{eq:EP}.
  
Pick a maximal multiplicatively independent subset of characteristic
roots.  Without loss of generality we can write this set as
$\{\lambda_1,\ldots,\lambda_t\}$ for some $t\leq s$.  As discussed in
Section~\ref{sec:background}, given the characteristic polynomial of
the recurrence, we can compute such a set, as well as integers $m_i$ and $n_{i,j}$
for $i \in \{1,\ldots,s\}$ and $j\in \{1,\ldots,t\}$, where 
the $m_i$ are non-zero, such that for all $i \in \{1,\ldots,s\}$ we have
\[ \lambda_i^{m_i} = \lambda_1^{n_{i,1}} \cdots \lambda_t^{n_{i,t}} \,
  .\]
Raising the left- and right-hand sides in the above equation to
  the 
  power $L$ and then taking logs, we get that
  \[ \log \lambda^L_i = \frac{n_{i,1}}{m_i} \log \lambda^L_1 + \cdots
    + \frac{n_{i,t}}{m_i} \log \lambda^L_t \] for all
  $i\in\{1,\ldots,s\}$.  In other words, for all
  $i\in\{1,\ldots,s\}$ we have that
  $\log \lambda^L_i=\ell_i( \log \lambda^L_{1},\ldots,\log
  \lambda^L_t)$ for an effectively computable linear form $\ell_i(X_1,\ldots,X_t)$ with rational
  coefficients.

  For $j\in\mathbb{N}$, define $F_j \in \mathbb{K}[X_1,\ldots,X_t]$ by
  \[ F_j(X_1,\ldots,X_t) := \frac{1}{j!} \sum_{i=1}^s \sum_{k=0}^j
    \binom{j}{k} L^k Q_i^{(k)}(0) \ell_i(X_1,\ldots,X_t)^{j-k}  \, .\]
Then by Equation~\eqref{eq:polyF} we have
\begin{gather}
  a_j = F_j\left(\log \lambda_1^L,\ldots,\log \lambda_t^L\right) \, .
\label{eq:G}
\end{gather}
We claim that $a_j \neq 0$ if $F_j$ is not identically zero.  Since
the coefficients of $F_j$ are algebraic over $\mathbb{Q}$ and the set
$\{\log \lambda_1,\ldots,\log \lambda_t\}$ is linearly independent
over $\mathbb{Q}$, the claim follows immediately from
Equation~\ref{eq:G} and the weak $p$-adic Schanuel Conjecture
(Conjecture~\ref{con:schanuel}).

We can now use the following procedure to compute a positive integer
$M$ such that $u_{Mn} \neq 0$ for all $n \in \mathbb{Z}$:
\begin{enumerate}
\item Successively compute the polynomials $F_0,F_1,\ldots$ .
\item Let $j_0$ be the least index $j$ such that $F_j$ is not identically zero.
  Compute $\nu:=v_p(a_{j_0})$ using the
  series~\eqref{eq:power-series}.  The weak $p$-adic Schanuel Conjecture
  implies that $a_{j_0}\neq 0$ and hence the computation of
  $v_p(a_{j_0})$ terminates.
\item Set $M:=Lp^{\nu+1}$.  Applying Proposition~\ref{prop:zero}, we have 
$u_{Mn}\neq 0$ for all non-zero integers~$n$. 
\end{enumerate}

Note that $j_0$ is well defined in Line~2, since if all the $a_j$
were zero, then $\boldsymbol u$ would be the identically zero
sequence, contradicting our initial assumption.
\end{proof}

A couple of remarks about the proof of Proposition~\ref{prop:oracle}
are in order.
\begin{remark}
Observe that the $p$-adic Schanuel Conjecture is
only required for termination of the procedure described at the end of
the proof.  If the procedure terminates then it is certain that
$a_{j_0}$ is the first non-zero coefficient of the power
series~\eqref{eq:power-series} and hence the outputted value of $M$ is
guaranteed to be such that $u_{Mn}\neq 0$ for all non-zero integers
$n$.
\end{remark}
\begin{remark}
  
  Examining the expression~\eqref{eq:polyF} and noting that
  $Q_i^{(k)}(0) =0$ for $k > \deg(Q_i)$, we see that the sequence
  $\langle j!a_j \rangle_{j=0}^\infty$ is given by an
  exponential polynomial corresponding to a (non-rational) LRS of order $d$ and
  hence at least one of $a_0,a_1,\ldots,a_{d-1}$ is non-zero.   This means that
  the index $j_0$ in Line~2 of the above procedure is at most $d-1$.
\end{remark}

\begin{theorem}
  Assuming the weak $p$-adic Schanuel Conjecture, there is a Turing
  reduction from the Skolem Problem to the Bi-Skolem Problem.
  \label{thm:Turing-Reduce}
  \end{theorem}
  \begin{proof}
    We present a recursive procedure that uses an oracle for the
    Bi-Skolem Problem to compute all the zeros of a non-degenerate 
    LRBS that is not identically zero.

    Given a non-degenerate LRBS
    $\boldsymbol{u}=\langle u_n \rangle_{n=-\infty}^\infty$, we use
    the oracle for the Bi-Skolem Problem to determine whether there
    exists $n\in \mathbb{Z}$ with $u_n=0$.  If the oracle outputs that
    no such $n$ exists then the procedure terminates.  Otherwise one
    searches for $n_0 \in \mathbb{Z}$ such that $u_{n_0}=0$; clearly
    the search is guaranteed to terminate.  Having found $n_0$, by
    reindexing the sequence $\boldsymbol{u}$ we can suppose that
    $n_0=0$.  Now we use Proposition~\ref{prop:zero} to compute a
    positive integer $M$ such that $u_{Mn}\neq 0$ for all $n\neq 0$.
    We then divide the sequence $\boldsymbol{u}$ into $M$ subsequences
    $\boldsymbol{u}^{(0)},\ldots,\boldsymbol{u}^{(M-1)}$, where for
    $j\in \{0,\ldots,M-1\}$, the $j$-th subsequence is given by
    $u^{(j)}_n = u_{Mn+j}$ for all $n\in \mathbb{N}$.  We know that
    $n=0$ is the only zero of $\boldsymbol{u}^{(0)}$.  We now
    recursively call the procedure to find all zeros of the remaining
    subsequences $\boldsymbol{u}^{(1)},\ldots,\boldsymbol{u}^{(M-1)}$.
    Observe that the computation must terminate since each recursive
    call involves discovering a new zero of the original sequence
    $\boldsymbol{u}$, and by the version of the Skolem-Mahler-Lech Theorem 
    for LRBS, there are only finitely many such zeros.
  \end{proof}


  If we restrict to recurrences of order at most $5$ then we obtain an
unconditional version of Theorem~\ref{thm:Turing-Reduce}.
\begin{theorem}
  \label{thm:order-5}
  There is a Turing reduction from the Skolem Problem for LRS of order
  at most $5$ to the Bi-Skolem Problem for simple LRBS of order at most
  $5$.
\end{theorem}  
\begin{proof}
  As summarised in Appendix~\ref{sec:skolem-hard}, the Skolem Problem
can be decided
  for all LRS $\langle u_n \rangle_{n=0}^\infty$
  of order at most $5$ except those that (after scaling) have a closed
  form $u_n = \sum_{i=1}^5 \alpha_i \lambda_i^n$ satisfying the
  following three conditions, where $\lambda_1,\ldots,\lambda_5$ lie in
  the ring of integers $\mathcal{O}_{\mathbb{K}}$ of a number field
  $\mathbb{K}$:
  \begin{enumerate}
    \item $\alpha_1 \neq - \alpha_3$;
          \item $\lambda_1\lambda_2=\lambda_3\lambda_4$;
          \item there is a prime ideal  $\mathfrak{p}$ in $\mathcal{O}_{\mathbb{K}}$  that
            divides $\lambda_1$ and $\lambda_3$ but not
            $\lambda_2,\lambda_4,\lambda_5$.
\end{enumerate}

The theorem at hand is proven using the procedure described in the
proof of Theorem~\ref{thm:Turing-Reduce}, which uses as a subroutine
the procedure described in Proposition~\ref{prop:oracle}.  To avoid
relying on the weak $p$-adic Schanuel Conjecture, it suffices to give
an unconditional proof of the termination of the latter procedure when
invoked on LRBS whose closed-form representation satisfies the above
conditions.  In other words, we must show that for such LRBS one can
compute a positive integer $M$ such that $u_{Mn}\neq 0$ for all
$n\in\mathbb{Z}$.

Let $p$ be a prime such that there is an embedding of $\mathbb{K}$
into $\mathbb{Q}_p$.  Recall from Section~\ref{sec:power} that there
exists a positive integer $L$ such that $u_{Ln}=f(n)$ for a $p$-adic
power series $f(X)=\sum_{j=0}^\infty a_j X^j$ such that
$a_1 = \sum_{i=1}^5 \alpha_i \log \lambda^L_i$.  The termination of
the procedure described in the proof of Proposition~\ref{prop:oracle}
will be assured if $a_1\neq 0$.
We claim that for an LRBS satisfying the above three conditions, one 
has $a_1=\sum_{i=1}^5 \alpha_i \log \lambda^L_i \neq 0$.

To prove the claim,
suppose for a contradiction that
$\sum_{i=1}^5 \alpha_i \log \lambda^L_i = 0$.  Raising to the $L$-th
power and then taking logarithms in Condition~2 above, we also have
$\log\lambda^L_1+\log\lambda^L_2-\log\lambda^L_3-\log\lambda^L_4=0$.
Combining the two previous equations to cancel $\log \lambda_1^L$ we
have
\begin{gather} (\alpha_2-\alpha_1)\log \lambda^L_2+(\alpha_3+\alpha_1)\log \lambda^L_3+
  (\alpha_4+\alpha_1)\log \lambda^L_4 + \alpha_5 \log\lambda^L_5=0 \, . 
\label{eq:lin-form}
\end{gather}

From Condition~1 ($\alpha_1\neq-\alpha_3$), we have that the
coefficient of $\log \lambda^L_3$ in Equation~\eqref{eq:lin-form} is
non-zero.  Applying Theorem~\ref{thm:p-adic-baker}, possibly several
times, we eventually obtain an equation
$\sum_{i=2}^5 \beta_i \log \lambda^L_i=0$ such that the $\beta_i$ are
integers and $\beta_3\neq 0$.  Equivalently, we have a multiplicative
relation among the characteristic roots that involves $\lambda_3$ but
not $\lambda_1$.  But this contradicts Condition~3 and the proof is
concluded.
\end{proof}

      \begin{theorem}
        \label{thm:decide-main}
        The Skolem Problem for simple LRS is decidable assuming the
        Skolem Conjecture and the weak $p$-adic Schanuel Conjecture.
        The Skolem Problem for LRS of order at most 5 is decidable
        assuming the Skolem Conjecture.
        \end{theorem}

        \begin{remark}
        Given that the Skolem Conjecture remains open in general, it is
        worth remarking that the proof of
        Theorem~\ref{thm:decide-main} sustains the following more
        general formulation: Consider a class $\mathcal{C}$ of simple
        LRBS that is closed under taking subsequences and under
        translations.  If the Skolem Conjecture holds for
        $\mathcal{C}$ then, assuming the weak $p$-adic Schanuel
        Conjecture, the Skolem Problem is decidable over LRS in
        $\mathcal{C}$.
\end{remark}

\section{The 
{\normalfont\bfseries\scshape Skolem}
 Tool
}
\label{tool}

We implemented our algorithm in the \textsc{Skolem} tool, which finds
all zeros (at both positive and negative indices) for simple
integer LRS, and produces independent certificates guaranteeing that there are no
further zeros. Even though we do not have complexity bounds,
\textsc{Skolem} can efficiently handle many interesting
examples, including several from the literature for which no
proof technique was previously known to apply. Our tool is available online at
\url{https://skolem.mpi-sws.org} and is accompanied by several
examples which can be experimented with on the spot.

The implementation is written in Python, using the SageMath
computer-algebra extension. This allows for the efficient and exact
manipulation of mathematical objects, including elements of
$\mathbb{Z}_p$.  Python handles integers of arbitrary sizes seamlessly,
making it especially suitable for our purposes, since even small
examples can give rise to very large numbers within the inner workings of
our algorithm. 


\begin{example}
Consider the LRS from \cite[Example 2.4]{LiptonLNOPW22}:
\[ u_{n+5} = 9u_{n+4} - 10u_{n+3} + 522u_{n+2} - 4745u_{n+1} +
  4225u_n \] with initial values (for $n=0,1,2,3,4$) of
$\langle -30, -27, 0, 469,$ $1762 \rangle$. It is shown
in~\cite{LiptonLNOPW22} to have a unique zero at index $2$ by being
non-zero modulo $12625$ at all indices larger than $2$. The
\textsc{Skolem} tool establishes this in a simpler way: after finding
$u_2 =0 $, the tool calculates that there are no zeros in $\langle u_{2+14n}\rangle_{n=-\infty}^\infty$
for $n\ne 0$. Then the tool computes that $\langle u_{k+14n}\rangle_{n=-\infty}^\infty$ is non-zero
modulo $29$ for each even $k\ne 2$, and non-zero modulo 2 for each odd
$k$ (where $0 \leq k \leq 13$). Observe that the computed modulo
classes, and thus the resulting certificate, is much smaller than
those arising from $12625$ as used in \cite{LiptonLNOPW22}.
\end{example}

\begin{example}
  Consider the LRS from \cite[Example 2.5]{LiptonLNOPW22}:
  \[ u_{n+6} =6u_{n+5} -26u_{n+4} + 66u_{n+3} -130u_{n+2} + 150u_{n+1}
    -125u_{n} \] with initial values (for $n=0,1,2,3,4,5$) of
  $\langle 0, 3, 11, -12, -125,$ $-177 \rangle$, which was established
  at the time of writing to lie beyond the reach of existing known techniques. The
  \textsc{Skolem} tool is able to show using the methods developed in
  the present paper that there are indeed no further zeros (other than $u_0=0$).
\end{example}

\begin{example}
  Consider the reversible order-$8$ LRS from \cite[Example
  3.5]{LiptonLNOPW22}:
  \[u_{n+8} = 6u_{n+7} -25u_{n+6} +66u_{n+5} -120u_{n+4} +
    150u_{n+3}-89u_{n+2}+18u_{n+1}-u_n \] with initial values (for
  $n = 0, \ldots, 7$) of $\langle 0,0,-48,-120,0,520,$
  $624,-2016 \rangle$, which likewise was established at the time to
  lie beyond the reach of existing techniques. \textsc{Skolem} shows
  that there are no zeros other than those lying at indices $0$, $1$,
  and $4$.
\end{example}

\subsection{Testing}

\begin{table}[t]
\begin{tabular}{@{}l|rrrrrr|rrrr@{}}

\toprule
     & \multicolumn{6}{c|}{Timeout of $60$ seconds} & \multicolumn{4}{c}{Timeout of $60\cdot \mathit{order}$ seconds} \\ \hline

 Order &  \multicolumn{1}{>{\raggedleft\arraybackslash}p{0.7cm}}{Total}& \multicolumn{1}{>{\raggedleft\arraybackslash}p{0.7cm}}{Success} & \multicolumn{1}{>{\raggedleft\arraybackslash}p{0.75cm}}{Degen-erate} & \multicolumn{1}{>{\raggedleft\arraybackslash}p{0.8cm}}{Not simple} &  \multicolumn{1}{>{\raggedleft\arraybackslash}p{0.85cm}}{Timeout} & \multicolumn{1}{>{\raggedleft\arraybackslash}p{1.1cm}|}{Timeout  \%} &        \multicolumn{1}{>{\raggedleft\arraybackslash}p{0.8cm}}{Total}          &   \multicolumn{1}{>{\raggedleft\arraybackslash}p{0.8cm}}{Success}     &   \multicolumn{1}{>{\raggedleft\arraybackslash}p{0.85cm}}{Timeout}     &    \multicolumn{1}{>{\raggedleft\arraybackslash}p{1.1cm}}{Timeout  \%}                         \\ \hline
2  & 9250 & 8836 & 358 & 50 & 0    & 0.00\%  & 1245 & 1200 & 0    & 0.00\%  \\
3  & 8995 & 8919 & 74  & 2  & 0    & 0.00\%  & 1327 & 1322 & 0    & 0.00\%  \\
4  & 9195 & 9157 & 35  & 2  & 1    & 0.01\%  & 1395 & 1392 & 0    & 0.00\%  \\
5  & 9188 & 8700 & 15  & 3  & 470  & 5.12\%  & 1303 & 1290 & 11   & 0.84\%  \\
6  & 9172 & 4905 & 10  & 6  & 4251 & 46.35\% & 1318 & 952  & 366  & 27.77\% \\ \hline
7  & 9213 & 1339 & 12  & 0  & 7862 & 85.34\% & 1328 & 310  & 1016 & 76.51\% \\
8  & 9157 & 378  & 10  & 0  & 8769 & 95.76\% & 1249 & 73   & 1173 & 93.92\% \\
9  & 9143 & 87   & 4   & 3  & 9049 & 98.97\% & 1330 & 18   & 1312 & 98.65\% \\
10 & 9047 & 25   & 8   & 1  & 9013 & 99.62\% & 1294 & 7    & 1286 & 99.38\% \\ \bottomrule
Total & 82360& 42346 & 526 & 67 & 39415 & 47.86\% & 11789  &6564 &   5164  &  43.80\% 
                  \\ \bottomrule
\end{tabular}
\caption{Table showing the distribution of outcomes by order. The line between orders 6 and 7 shows the boundary beyond which more than 50\% of runs timeout. The second experiment shows the timeout rate when the timeout is increased to $60s\cdot \mathit{order}$ (`degenerate' and `not-simple' counts omitted as the distribution is similar to the 60s timeout experiment and unaffected by the timeout). }
\label{table:successrate}

\end{table}
The \textsc{Skolem} tool was tested on a suite of random LRS, with the
order taken uniformly between $2$ and $10$ and the coefficients taken
uniformly at random between $-20$ and $20$. Tests were run for 48
hours using a $60$-second timeout\footnote{Testing was conducted using
  SageMath 9.5 in Docker on a Dell PowerEdge M620 blade equipped with
  2$\times$ 3.3 GHz Intel Xeon E5-2667 v2 (2$\times 8$ cores, 32 with
  hyper-threading) and 256GB ram. Testing was restricted to 16
  parallel threads (50\% of the computer's resources) for
  institutional reasons.}, generating 82367 test instances.\footnote{$7$
  instances were discarded: $6$ happened to be the zero
  sequence, one resulted in an exception (outside of the main tool code)
  which was later fixed.}

The results are presented in \cref{table:successrate,table:stats}. In
particular, from order $7$ onwards the tool is unable to handle more than half of
the instances within one minute, with the timeout percentage jumping
significantly from order 6. Both degenerate and non-simple LRS
instances are very sparse, and 
as expected the higher the order the fewer such instances were randomly produced.

The experiments were re-run using a timeout of
$60\cdot \mathit{order}$ seconds (i.e., ranging from 2--10 minutes) in
order to determine whether the $60$-second timeout was too strict. The
timeout percentage does decrease, but the overall pattern shows that
the vast majority of LRS of order $7$ and above could not be handled to completion before
the timeout.

\cref{table:stats} presents statistical information. In the main
experiment the average time is below 9 seconds for order-$6$ examples
that succeed within 60 seconds. However, the decrease from order-$8$
onwards is
explained by there being significantly fewer examples succeeding
within 60 seconds. On average there are very few
zeros (as can be expected) and those that do occur are almost
always those occurring within the initial value (the largest zero is
nearly always at index less than the order).

The average maximum jump step used (i.e., $M$ such that $\langle u_{Mn}\rangle$ has no
zeros for $n\ne 0$ and $u_0 = 0$), is observed on average to grow with the order
(except at order 10 with only 25 successful samples).

\begin{table}[]
\begin{tabular}{rrrrrrrr|r}
\hline
\multicolumn{1}{l}{Order} & \multicolumn{1}{>{\raggedleft\arraybackslash}p{1.15cm}}{mean time (seconds)} & \multicolumn{1}{>{\raggedleft\arraybackslash}p{1.15cm}}{mean count   of zeros} & \multicolumn{1}{>{\raggedleft\arraybackslash}p{1.15cm}}{max count of zeros} & \multicolumn{1}{>{\raggedleft\arraybackslash}p{1.15cm}}{mean max zero} & \multicolumn{1}{>{\raggedleft\arraybackslash}p{0.8cm}}{max zero index} & \multicolumn{1}{>{\raggedleft\arraybackslash}p{1.15cm}}{mean tree depth} & \multicolumn{1}{>{\raggedleft\arraybackslash}p{1.15cm}|}{mean max jump} & \multicolumn{1}{>{\raggedleft\arraybackslash}p{1.5cm}}{mean time (seconds) $60s\cdot \mathit{order}$ } \\ \hline
2 & 0.03  & 0.06 & 1 & 0.63 & 6 & 1.06 & 5.11   & 0.03   \\
3  & 0.05  & 0.08 & 3 & 1.03 & 5 & 1.08 & 13.56  & 0.06   \\
4  & 0.19  & 0.10 & 3 & 1.58 & 7 & 1.10 & 37.05  & 0.22   \\
5  & 4.82  & 0.11 & 2 & 2.05 & 6 & 1.11 & 107.39 & 10.07  \\
6  & 8.95  & 0.20 & 2 & 2.58 & 7 & 1.20 & 254.12 & 55.36  \\ \hline
7  & 11.72 & 0.30 & 2 & 2.80 & 9 & 1.30 & 482.34 & 70.19  \\
8  & 8.07  & 0.35 & 2 & 3.51 & 8 & 1.35 & 533.92 & 68.21  \\
9  & 7.38  & 0.38 & 1 & 4.24 & 8 & 1.38 & 689.33 & 138.40 \\
10 & 5.71  & 0.40 & 1 & 5.20 & 9 & 1.40 & 249.60 & 112.11\end{tabular}
\caption{Table listing statistical information for successful runs, by
  order. Line between orders 6 and 7 shows the boundary beyond which more than 50\% of runs timeout,
resulting in skewed analysis for the subsequent rows. For the second experiment, with timeout of $60\cdot \mathit{order}$ seconds, only the mean time is shown as there are fewer data points.
}
\label{table:stats}
\end{table}

\appendix

\section{Hard Cases of the Skolem Problem at Order 5}
\label{sec:skolem-hard}
      As explained in~\cite{LiptonLNOPW22}, the Skolem Problem is
      known to be decidable for all LRS of
      order at most 5 except for those sequences
      $\boldsymbol u = \langle u_n \rangle_{n=0}^\infty$ having an
      exponential-polynomial representation 
\begin{gather}
      u_n = \alpha_1 \lambda_1^n + \overline{\alpha_1}\overline{\lambda_1}^n
        + \alpha_2 \lambda_2^n + \overline{\alpha_2}\overline{\lambda_2}^n + \alpha_3
        \lambda_3^n 
\label{eq:SKOLEM5}
      \end{gather}
      such that
      $\alpha_1,\alpha_2,\alpha_3,\lambda_1,\lambda_2,\lambda_3 \in \overline{\mathbb{Q}}$
      satisfy $|\lambda_1|=|\lambda_2|>|\lambda_3|$ and
      $\lambda_1,\lambda_2,\lambda_3$ are not all units.  It is
      further shown in~\cite{LiptonLNOPW22} that by scaling sequences
      of this form we can assume that there exists a prime ideal
      $\mathfrak{p}$ in the ring of integers of the number field
      generated by $\lambda_1,\lambda_2,\lambda_3$ such that
      $\mathfrak{p}$ divides $\lambda_1$ and $\lambda_2$, but not
      $\overline{\lambda_1},\overline{\lambda_2}$ and $\lambda_3$.

      Here we make the further observation that for non-degenerate LRS
      of the form~\eqref{eq:SKOLEM5}, under the assumption that
      $|\alpha_1|=|\alpha_2|$, there is a computable upper bound on $n$ such that
      $u_n=0$.

      By scaling we can assume without loss of generality
      that $|\lambda_1|=|\lambda_2|=1$ and $|\alpha_1|=|\alpha_2|=1$.
      Thus we can write 
      $\lambda_1 = e^{i\theta_1}$ and $\lambda_2 = e^{i \theta_2}$ for
      $\theta_1,\theta_2 \in [0,2\pi)$ and we can put
      $\alpha_1 = e^{i \phi_1}$ and $\alpha_2 = e^{i \phi_2}$ for
      $\phi_1,\phi_2 \in [0,2\pi)$.  Then we have
\begin{eqnarray*}
  u_n &=&
  \alpha_1\lambda_1^n+\overline{\alpha_1}\overline{\lambda_1}^n+
          \alpha_2\lambda_2^n+\overline{\alpha_2}\overline{\lambda_2}^n
          + \alpha_3\lambda_3^n \\
                                  &=& 2 \cos(n\theta_1+\phi_1)+
                                      2 \cos(n\theta_2+\phi_2) + \alpha_3\lambda_3^n \\
&=&   4\left(\cos\left(\frac{n(\theta_1+\theta_2)+\phi_1+\phi_2}{2}\right) 
    \cos\left(\frac{n(\theta_1-\theta_2)+\phi_1-\phi_2}{2}\right)\right)
    +\alpha_3\lambda_3^n 
                                    \, .                                      
\end{eqnarray*}

By non-degeneracy of $\boldsymbol u$, the terms
$\cos\left(\frac{n(\theta_1+\theta_2)+\phi_1+\phi_2}{2}\right)$ and
$\cos\left(\frac{n(\theta_1-\theta_2)+\phi_1-\phi_2}{2}\right)$
are respectively zero for at most one value of $n \in \mathbb{N}$.
Furthermore, using Baker's Theorem on linear forms in logarithms
(see~\cite{MST84,Ver85} for details), each of these terms has
a lower bound (when non-zero) of the form $\frac{c}{n^d}$ for
explicitly computable constants $c$ and $d$.  Since $|\lambda_3|<1$ it
follows that $u_n \neq 0$ for all $n \geq n_0$ for some effective
threshold $n_0$.

\bibliography{refs}
\end{document}